\documentclass[letterpaper, 10 pt, conference]{ieeeconf}

\IEEEoverridecommandlockouts                              

\usepackage{tikz}
\usepackage{textcomp}

\usetikzlibrary{matrix}
\usetikzlibrary{angles,quotes}

\usepackage{graphics,graphicx} 
\usepackage{epsfig} 
\usepackage{mathptmx} 
\usepackage{times} 
\usepackage{amsmath} 
\usepackage{amssymb,nth, mathtools,dsfont}  
\usepackage{amsbsy}
\usepackage{epstopdf}
\usepackage[noadjust]{cite}
\usepackage{romannum}
\usepackage{xcolor}
\usepackage{subcaption}
\usepackage{multirow}
\usepackage{comment}
\usepackage[normalem]{ulem}

\newcommand{\bmtx}{\begin{bmatrix}}
\newcommand{\emtx}{\end{bmatrix}}
\newcommand{\bsmtx}{\left[ \begin{smallmatrix}} 
\newcommand{\esmtx}{\end{smallmatrix} \right]}
\newcommand{\field}[1]{\mathbb{#1}}
\newcommand{\R}{\field{R}}

\newcommand{\N}{\field{N}}

\newcommand{\Sym}{\field{S}}

\newcommand{\svn}[1]{{\color{blue}{(SVN: #1)}}}

\newtheorem{theorem}{Theorem}

\newtheorem{lemma}{Lemma}

\newtheorem{remark}{Remark}
\newtheorem{defin}{Definition}

\newtheorem{assumption}{Assumption}

\title{Data-Driven Stability and Performance Analysis of Lurye Systems}

\author{Sahel Vahedi Noori and Peter Seiler
	\thanks{S. Vahedi Noori and P. Seiler are with the Department of Electrical Engineering \& Computer Science at the University of Michigan ({\tt\small sahelvn@umich.edu} and
		{\tt\small pseiler@umich.edu}).
  }
	}

\begin{document}

\maketitle

\begin{abstract}

This paper develops data-driven conditions for certifying internal stability and induced-$\ell_2$ performance of discrete-time Lurye systems. The Lurye system is an interconnection of a nominal LTI system in feedback with a static, memoryless nonlinearity.  The nonlinearity satisfies a known set of quadratic constraints on the inputs and outputs. Existing conditions for Lurye systems require a state-space realization of the nominal LTI dynamics. Our first data-driven result instead formulates the stability and performance conditions using finite input, output, and state trajectories of the nominal LTI block. Our second condition removes the need for measured state trajectories by reconstructing the state sequence, up to a similarity transformation, from input/output data. This state reconstruction is performed using deterministic  subspace-identification techniques. Both data-driven conditions are expressed as convex semidefinite programs. These conditions, given sufficiently exciting inputs, recover the corresponding model-based Lurye condition in the noiseless setting.  The proposed methods are illustrated via a simple example with a sector-bounded  nonlinearity.  Both proposed methods obtain the same induced-$\ell_2$ gain bound as the model-based approach while using only trajectory data from the nominal system.

\end{abstract}



\section{Introduction}

The analysis and control of dynamical systems in the presence of uncertainty and nonlinearities is a central problem in control theory. A widely studied class of systems is given by Lurye-type interconnections, consisting of a linear time-invariant (LTI) system in feedback with a nonlinear operator. Classical approaches to analyzing such systems rely on tools such as dissipativity theory and integral quadratic constraints (IQCs) \cite{megretski97,seiler15,veenman2013stability}. This leads to convex conditions for stability and performance in terms of linear matrix inequalities (LMIs) \cite{Willems:71, Willems:1968, Carrasco:2016, carrasco19, safonov2000zames, kulkarni2002all, zhang2021lyapunov, scherer2022dissipativity, veenman16, fetzer17, megretski97, seiler15}. These methods require explicit knowledge of the system model and characterize nonlinearities through quadratic constraints (QCs) on their input-output behavior.

In recent years, there has been a growing interest in data-driven control, where system-theoretic properties and control laws are inferred directly from data without explicitly identifying a model. This paradigm is motivated by the increasing complexity of modern systems, where accurate modeling is often difficult or impractical. A fundamental result in this area is Willems’ fundamental lemma \cite{WILLEMS2005325}, which shows that the behavior of an unknown LTI system can be fully characterized from a single persistently exciting (PE) trajectory. Building on this result, a large body of work has emerged on direct data-driven analysis and control, including data-enabled predictive control and behavioral approaches to optimal control and stabilization \cite{markovsky2023data, coulson2019data, coulson2019regularized, de2019formulas, van2020data, van2022tutorial, farjadnasab2022model, noori2025data,vanwaarde26}.

An important line of research concerns data-driven dissipativity analysis. Dissipativity provides a unifying framework for stability and performance analysis through energy-like inequalities. Recent works have sought to extend these concepts to the data-driven setting. The notion of data informativity was introduced
in~\cite{van2020data}. This provides necessary and sufficient conditions under which system-theoretic properties can be inferred from data. Building on this framework, data-driven conditions for dissipativity and control have been developed using
LMI formulations. For instance, \cite{van2022dissipativity} derives conditions under which dissipativity can be certified directly from measured data, including in the presence of noise. Other recent efforts have explored finite-horizon dissipativity and performance analysis using input-output data and behavioral representations~\cite{wieler2021data, van2023behavioral, koch2021provably}. 


It is natural to extend these data-driven results to analyze Lurye-type interconnections. This enables analysis of systems with nonlinearities.  The standard (model-based) approach requires explicit models for the nominal LTI system and known IQCs for the nonlinearity. Data-driven methods have been used to learn IQCs for nonlinearities (while still requiring a model of the nominal LTI system) \cite{burgin2025robust,burgin2026robust,gupta2023data,gupta2026non}. In contrast, there is limited work for the case where the IQC is known but data is used for the nominal model. The most relevant work in this setting is  \cite{luppi2022data}.  This paper assumes the nonlinearity satisfies a known (and fixed) QC and input/state data from the nominal plant is available.  A convex, LMI condition is given to design a stabilizing state-feedback controller.

Our paper contributes to this literature by developing data-driven conditions for stability and induced-$\ell_2$ performance for discrete-time Lurye systems.  We assume the nonlinearity satisfies a known set of QCs, but we do not assume that a model for the nominal LTI system is known. Instead, we formulate two data-driven conditions to assess stability and performance.  
These conditions are expressed as convex semidefinite programs. The first condition  uses input/output/state data (Section~\ref{sec:dataDriven}) while the second condition only uses input/output data (Section~\ref{sec:IOdataDriven}). The latter leverages deterministic subspace identification techniques \cite{overschee1996subspace, van1997numerical, VERHAEGEN01111992, van1994n4sid} to reconstruct a state sequence from measured input-output data. The reconstructed state is equivalent to the true state up to an invertible change of coordinates. 
A simple example is given 
in Section~\ref{sec:example} to illustrate both data-driven stability and performance conditions.

Our analysis conditions provide three novel aspects compared to the complementary synthesis conditions in  \cite{luppi2022data}.  First, our formulation certifies both internal stability and induced-$\ell_2$ performance.  Second, we develop an input-output formulation that does not
require state measurements (Section~\ref{sec:IOdataDriven}). Lastly, the QC matrix is not  fixed: we optimize over a convex set of valid QCs jointly with the storage matrix. This optimization is a convex SDP yields the best  (smallest) bound on the induced-$\ell_2$ gain. Extending this analysis condition to the
synthesis formulation in~\cite{luppi2022data} yields, in general, a nonconvex formulation due to couplings in the storage matrix, controller and QC decision variables.

\section{Notation}
\label{sec:notation}

Let $\R^n$ and $\R^{n\times m}$ denote the sets of real $n\times 1$
vectors and $n\times m$ matrices, respectively.  $\Sym^n$ is the set of real symmetric, $n\times n$ matrices. A matrix $M\in \R^{n\times n}$ is \emph{doubly hyperdominant} if the off-diagonal elements are non-positive, and all row and column sums are non-negative. $M^\dagger$ denotes the Moore--Penrose pseudoinverse of $M$. 


Next, let $\N$ denote the set of non-negative integers.  Let $v:\N \to \R^n$ and $w:\N \to \R^n$ be real, vector-valued sequences.  Define the inner product $\langle v,w \rangle : = \sum_{k=0}^\infty v(k)^\top w(k)$.
A sequence $v$ is said to be in $\ell_2^n$ if $\langle v,v\rangle <\infty$. In addition, the $2$-norm for $v \in \ell_2^n$ is defined as
$\|v\|_2:=\sqrt{ \langle v,v\rangle}$.
We will use $\ell_{2e}$ to denote the extended space of sequences whose finite-time truncations belong to $\ell_{2}$.

\section{Problem statement}

Consider the interconnection shown in Figure~\ref{fig:LFTdiagram}, denoted as $F_U(G,\Delta_\Phi)$. This has a memoryless nonlinearity $\Delta_\Phi$  in feedback around the top channels
of a nominal system $G$. 

\begin{figure}[h!t]
\centering
\scalebox{0.92}{
\begin{picture}(100,75)(40,25)
 \thicklines
 \put(75,25){\framebox(40,40){$G$}}
 \put(139,40){$d$}
 \put(150,35){\vector(-1,0){35}}  
 \put(46,40){$e$}
 \put(75,35){\vector(-1,0){35}}  
  \put(78,73){\framebox(34,34){$\Delta_\Phi$}}
 \put(46,70){$v$}
 \put(55,55){\line(1,0){20}}  
 \put(55,55){\line(0,1){35}}  
 \put(55,90){\vector(1,0){23}}  
 \put(139,70){$w$}
 \put(135,90){\line(-1,0){23}}  
 \put(135,55){\line(0,1){35}}  
 \put(135,55){\vector(-1,0){20}}  
\end{picture}
}
\caption{Interconnection $F_U(G,\Delta_\Phi)$ of a nominal discrete-time LTI system $G$ and a memoryless nonlinearity $\Delta_\Phi$.}
\label{fig:LFTdiagram}
\end{figure}
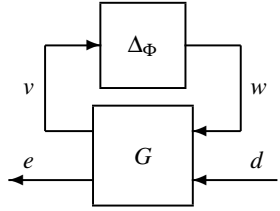

The nominal part $G$ is a discrete-time, linear time-invariant (LTI)
system described by the following state-space model:
\begin{align}
  \label{eq:LTInom}
  \begin{split}
    x(k+1) & = A\, x(k) + B_1\,w(k) +  B_2\, d(k) \\
    v(k) & =C_{1}\,x(k)+D_{11}\, w(k)+ D_{12} \,d(k)\\
    e(k) & =C_{2}\,x(k)+D_{21}\, w(k)+ D_{22}\,d(k),
  \end{split}
\end{align}
where $x(k) \in \R^{n_x}$ is the state at time $k$. Similarly, the inputs at time $k$ are $w(k) \in \R^{m}$ and $d(k)\in \R^{n_d}$, while the outputs at time $k$ are $v(k) \in \R^{m}$ and $e(k)\in \R^{n_e}$.  The  nonlinearity $\Delta_\Phi:\ell_{2e}^m\to \ell_{2e}^m$ is defined by applying a memoryless function $\Phi:\R^m \to \R^m$ at each point in time. In other words, $\Delta_\Phi$ maps $v\in \ell_{2e}^m$ to $w=\Delta_\Phi(v)\in\ell_{2e}^m$ by $w(k)=\Phi( v(k) )$ at each time $k$.
The interconnection $F_U(G,\Delta_\Phi)$ is known as a linear fractional transformation (LFT)  \cite{zhou96}. 

This feedback interconnection involves an implicit equation if $D_{11}\ne 0$.  Specifically, the second equation in \eqref{eq:LTInom} combined with $w(k)=\Phi( v(k) )$ yields:
\begin{align}
   \label{eq:WellPosed}
   v(k)=C_{1}\,x(k)+D_{11}\, \Phi( v(k) )+ D_{12} \,d(k).
\end{align}
This algebraic equation is \emph{well-posed} if there exists a unique solution $v(k)$ for all  values of $x(k)$ and $d(k)$.  Well-posedness of this equation implies that the dynamic system $F_U(G,\Delta_\Phi)$ is well-posed in the following sense:  
\vspace{0.1in}
\begin{defin}
\label{def:wellposed}
The interconnection $F_U(G,\Delta_\Phi)$ is \uline{well-posed} if for all
initial conditions $x(0)\in\R^{n_x}$ and inputs $d\in \ell_{2e}^{n_d}$ 
there exist unique solutions $x\in \ell_{2e}^{n_x}$, $e\in \ell_{2e}^{n_e}$ and $w, v \in \ell_{2e}^{m}$ to the system $F_U(G,\Delta_\Phi)$. 
\end{defin}
\vspace{0.1in}

There are simple sufficient conditions for well-posedness, e.g.,
$D_{11}=0$, which removes the algebraic self-reference of $v(k)$ in
\eqref{eq:WellPosed}. Less conservative conditions are 
given in \cite{valmorbida18} and \cite{zaccarian02}. For simplicity we assume well-posedness in our main results.


Our goal in this paper is to certify stability and performance of the interconnection $F_U(G,\Delta_\Phi)$ using only measured data from the nominal system $G$ and known input-output bounds on the nonlinearity $\Delta_\Phi$. Specifically, we are concerned with internal stability and finite induced-$\ell_2$ gain of the interconnection as defined next.

\vspace{0.1in}
\begin{defin}
A well-posed interconnection $F_U(G,\Delta_\Phi)$ is \uline{internally stable} if $x(k)\to 0$ from any initial condition $x(0)$ with $d(k)=0$ for $k\in \N$. In other words, $F_U(G,\Delta_\Phi)$ is internally stable if $x=0$ is a globally asymptotically stable equilibrium point with no external input.
\end{defin}
\vspace{0.1in}
\begin{defin}
 A well-posed interconnection $F_U(G,\Delta_\Phi)$ has \uline{finite induced-$\ell_2$ gain} if there exists  $\gamma<\infty$ such that the output $e$ generated by any $d\in \ell_{2e}^{n_d}$ with $x(0)=0$ satisfies $\|e\|_2 \le \gamma \, \|d\|_2$.  The infimum of all such bounds on the input-output gain is denoted by $\|F_U(G,\Delta_\Phi)\|_{2\to2}$.
\end{defin}


\section{Preliminary Results}

This section reviews an existing sufficient condition for internal stability and finite induced $\ell_2$-gain for the interconnection $F_U(G,\Delta_\Phi)$. This existing condition requires a model for $G$, i.e. it assumes the state matrices in \eqref{eq:LTInom} are known.  It also requires known bounds on the input-output behavior of the nonlinearity $\Delta_\Phi$.  These bounds are given in terms of quadratic constraints as defined next.

\vspace{0.1in}
\begin{defin}
A function $\Phi:\R^{m} \to \R^{m}$ satisfies the \uline{Quadratic Constraint (QC)} defined by $M\in \Sym^{2m}$ if the
following inequality holds for all $v\in \R^{m}$:
\begin{align*}
\bmtx v \\ \Phi(v)\emtx^\top 
M
\bmtx v \\ \Phi(v)\emtx \ge 0 .
\end{align*}
\end{defin}
\vspace{0.1in}

Next, we state the sufficient condition for stability and performance. This condition is stated using the matrix function $L_{MB}: \Sym^{n_x} \times \Sym^{2m} \times \R_{>0} \times \R_{>0} \to \Sym^{n_x+m+n_d}$ defined as follows:
\begin{align}
\nonumber
& L_{MB}(P, M, \gamma^2, \epsilon) := \epsilon I- \bmtx P & 0 & 0 \\ 0 & 0 & 0 \\ 0 & 0 & \gamma^2I \emtx
 + \bmtx A^\top \\ B_{1}^\top  \\ B_{2}^\top \emtx
P
  \bmtx A^\top \\ B_{1}^\top  \\ B_{2}^\top \emtx^\top \\
\label{eq:LMB}
& \hspace{0.4in}
+ \bmtx C_2^\top \\ D_{21}^\top \\ D_{22}^\top\emtx
  \bmtx C_2^\top \\ D_{21}^\top \\ D_{22}^\top \emtx^\top
+  \bmtx C_1^\top & 0\\ D_{11}^\top & I\\ D_{12}^\top & 0\emtx
M
  \bmtx C_1^\top & 0\\ D_{11}^\top & I\\ D_{12}^\top & 0\emtx^\top
\end{align}
This matrix function depends on the state matrices of $G$. The subscript $MB$ indicates that this will be part of the model-based condition. Specifically, the lemma below uses $L_{MB}$ to formulate a sufficient condition for internal stability and finite induced-$\ell_2$ gain of the interconnection $F_U(G,\Delta_\Phi)$. The proof is based on standard arguments using Lyapunov/dissipativity theory combined with the QC for the nonlinearity \cite{megretski97, seiler15, veenman16}. The proof is given as it provides the starting point for our main results in the next section.

\vspace{0.1in}
\begin{lemma}
    \label{lem:stabMB}
    Consider the interconnection $F_U(G,\Delta_\Phi)$ where
    $G$ is the LTI system \eqref{eq:LTInom} and $\Phi:\R^m \to \R^m$ is a memoryless nonlinearity that satisfies the QC defined by $M$.  Assume the interconnection is well-posed as in Definition~\ref{def:wellposed}.
    Then the interconnection $F_U(G,\Delta_\Phi)$ is internally stable and has $\|F_U(G,\Delta_\Phi)\|_{2\to 2} < \gamma$ if there exists $P\succ 0$, $\gamma > 0$, and $\epsilon>0$ such that $L_{MB}(P,M,\gamma^2,\epsilon) \preceq 0$.
\end{lemma}
\begin{proof} 
Define the storage function $V\left(x\right) := x^\top P x$. 
By well-posedness, the LFT $F_U(G,\Delta_\Phi)$ has a unique solution $x\in \ell_{2e}^{n_x}$, $e\in \ell_{2e}^{n_e}$ and $w, v \in \ell_{2e}^{m}$ for all
initial conditions $x(0)\in\R^{n_x}$ and inputs $d\in \ell_{2e}^{n_d}$. 
Left and right multiply $L_{MB}(P,M,\gamma^2,\epsilon)$, defined in~\eqref{eq:LMB}, by $[x(k)^\top \; w(k)^\top \; d(k)^\top]$ and its transpose.  The result gives the following dissipation inequality:
\begin{align*}
    \begin{split}  
    & V(x(k+1)) - V(x(k)) + \epsilon x(k)^\top x(k) + e(k)^\top e(k) \\
    & + \epsilon w(k)^\top w(k) + \bmtx v(k)\\w(k)\emtx^\top M \bmtx v(k)\\w(k)\emtx \leq (\gamma^2-\epsilon)d(k)^\top d(k) 
    \end{split}
\end{align*}
The term involving $M$ is non-negative for every $k$ as the nonlinearity satisfies the QC defined by $M$. Moreover, the term $\epsilon w(k)^\top w(k)$ is also non-negative at every $k$.  Hence the inequality still holds after these terms are dropped. Summing the remaining inequality from $k = 0$ to an arbitrary time $k = T_0 \ge 0$ yields
\begin{align}
\label{eq:sumDissipation}
    \begin{split}
        & \epsilon \sum_{k=0}^{T_0} x(k)^\top x(k) + \sum_{k=0}^{T_0} e(k)^\top e(k) \\
        & \leq V(x(0)) - V(x(T_0+1))
        +(\gamma^2-\epsilon) \sum_{k=0}^{T_0} d(k)^\top d(k).
    \end{split}
\end{align}

First, consider $d(k)=0$ for all $k$. Inequality \eqref{eq:sumDissipation} 
combined with $V(x(T_0+1))\ge 0$ (because $P\succ 0$) gives:
\begin{align*}
    \epsilon \sum_{k=0}^{T_0} x(k)^\top x(k) \le V(x(0))
\end{align*}
This bound holds for any $T_0 \ge 0$. Thus the left side is uniformly bounded as $T_0\to  \infty$ since $V(x(0))$ is finite. Hence $x\in \ell_2^{n_x}$ and $\|x\|_2^2 \le \frac{1}{\epsilon} V(x(0))$. This implies $x(k) \to 0$ as $k\to \infty$. We conclude that the interconnection $F_U(G,\Delta_\Phi)$ is internally stable.

Now assume $x(0)=0$ and $d \in \ell_{2}^{n_d}$. Inequality \eqref{eq:sumDissipation}
combined with $V(x(T_0+1))\ge 0$ and $V(x(0))=0$
yields:
\begin{align*}
    \sum_{k=0}^{T_0} e(k)^\top e(k) \leq (\gamma^2-\epsilon) \sum_{k=0}^{T_0} d(k)^\top d(k)
\end{align*}
Taking the limit as $T_0\to \infty$ on the right side yields:
\begin{align*}
  \sum_{k=0}^{T_0} e(k)^\top e(k) \le (\gamma^2-\epsilon) \| d\|_2^2    
\end{align*}
This implies the left side is uniformly bounded as $T_0\to \infty$ and hence $e\in \ell_2^{n_e}$. Moreover,  $\| e\|_2^2 \le (\gamma^2-\epsilon) \| d\|_2^2$. Therefore, the interconnection has finite induced-$\ell_2$ gain with $\|F_U(G,\Delta_\Phi)\|_{2\to 2} < \gamma$.
\end{proof}
\vspace{0.1in}

This theorem provides a sufficient condition to bound the $\ell_2$ gain of the system. Suppose the nonlinearity satisfies any QC defined by a matrix $M$ in the set $\mathcal{M}\subset \Sym^{2m}$. We can solve the following optimization to find the best (smallest) bound on the gain for a fixed $\epsilon>0$:
\begin{align}
\label{eq:IncSDP}
    \min_{M\in \mathcal{M}, P\succ 0, \gamma^2>0} \gamma^2
    \mbox{ s.t. } L_\text{MB}(P,M,\gamma^2,\epsilon) \preceq 0.
\end{align}
The constraint $ L_\text{MB}(P,M,\gamma^2,\epsilon) \preceq 0$ is a linear matrix inequality (LMI) in $P$, $M$ and $\gamma^2$. Moreover, the set $\mathcal{M}$ if often described by affine constraints on $M$. In this case, the optimization \eqref{eq:IncSDP} is a (convex) semidefinite program (SDP).  The data driven conditions in the next section can also be formulated as SDPs.

\section{Main Results}

\subsection{Data-Driven Stability and Performance}
\label{sec:dataDriven}

In this subsection, we derive a data-driven stability and performance condition.  This condition assumes that measurements are available for the state $x$, inputs $(w,d)$ and outputs $(v,e)$.  We will extend our condition in the next subsection to only require input-output data, and not state data.

First, define the following stacked data matrices for any signals $(x,w,d,v,e)$ with data length $N\ge 1$:
\begin{align}
\begin{split}
    X_{N}&= \bmtx x(0) & x(1) & \cdots & x(N-1) \emtx, \\
    X_{N}^{+}&= \bmtx x(1) & x(2) & \cdots & x(N) \emtx, \\
    W_{N}&= \bmtx w(0) & w(1) & \cdots & w(N-1) \emtx, \\
    D_{N}&= \bmtx d(0) & d(1) & \cdots & d(N-1) \emtx, \\
    V_{N}&= \bmtx v(0) & v(1) & \cdots & v(N-1) \emtx, \\
    E_{N}&= \bmtx e(0) & e(1) & \cdots & e(N-1) \emtx.
    \label{eq:data}
\end{split}
\end{align}
Note that each stacked data matrix has $N$ columns, and the number of rows equals the dimension of the corresponding signal. We require the input data to be sufficiently rich, as formalized next.

The LFT system $G$, given in \eqref{eq:LTInom}, has the following stacked input-output signals:
\begin{align*}
  u:= \bmtx w \\ d \emtx
  \mbox{ and }   
  y:=\bmtx v \\ e \emtx.
\end{align*}
The dimensions of these stacked input/output signals are $n_u:=m+n_d$ and $n_y:=m+n_e$, respectively. The lumped realization of $G$, from $u$ to $y$, is denoted as $(A,B,C,D)$ where $B:=[B_1, \, B_2]$ and $(C,D)$ are defined similarly. Let input data  $\{u(k)\}_{k=0}^{N-1}$
be given along with integers $i,j \ge 1$ satisfying $N=i+j-1$. Following Section~2.1.2 of~\cite{overschee1996subspace}, 
define a  block-Hankel matrix from the input data
\begin{align}
    U_{i,j}^p := \bmtx u(0) & u(1) & \cdots & u(j-1)\\ u(1) & u(2) & \cdots & u(j)\\ \vdots & \vdots & & \vdots\\ u(i-1) & u(i) & \cdots & u(i+j-2)
    \emtx 
    \in \R^{(in_u) \times j}.
    \label{eq:pInHank}
\end{align}
The subscripts denote that the matrix has $i$ block rows (each of dimension $n_u$) and $j$ columns. The lower right entry corresponds to $u(N-1)$. Hence this block-Hankel matrix  requires input data of length $N$. It is used to define the notion of persistency of excitation from~\cite{WILLEMS2005325}.

\vspace{0.1in}
\begin{defin}
\label{def:PE}
The finite input sequence $\{u(k)\}_{k=0}^{N-1}$ is
\uline{persistently exciting (PE) of order $i$} if
$U_{i,j}^p$ has full row rank, i.e., $\operatorname{rank}\big(U_{i,j}^p\big)=in_u.$
\end{defin}
\vspace{0.1in}

Next, define the following matrix function $L_{SSD}: \Sym^{n_x} \times \Sym^{2m} \times \R_{>0} \times \R_{>0} \to \Sym^{N}$
using the data matrices in \eqref{eq:data}:
\begin{align}
\nonumber
& L_{SSD}(P,M,\gamma^2,\epsilon):= \,\epsilon \bmtx X_{N} \\ W_{N} \\ D_{N} \emtx^\top \bmtx X_{N} \\ W_{N} \\ D_{N} \emtx- X_{N}^\top\, P\, X_{N} + {X_{N}^{+}}^\top\, P\, X_{N}^{+}\\
& 
\hspace{0.8in}
- \gamma^2\, D_{N}^\top\, D_{N} + E_{N}^\top\, E_{N}
+  \bmtx V_{N} \\ W_{N}\emtx^\top
M
  \bmtx V_{N} \\ W_{N}\emtx
\label{eq:LSSD}
\end{align}
The subscript $SSD$ indicates that this will be part of the state-space data-driven formulation, which uses measured state trajectories in addition to input-output data.
The following theorem is obtained by performing a congruence transformation on the model-based function $L_{MB}$ used in Lemma~\ref{lem:stabMB}.

\vspace{0.1in}
\begin{theorem}
\label{thm:stabMF}
Consider the interconnection $F_U(G,\Delta_\Phi)$, where $G$ is the LTI system in \eqref{eq:LTInom}, $(A,B)$ is controllable, and $\Phi:\R^m \to \R^m$ is a memoryless nonlinearity that satisfies the QC defined by $M$. Assume the interconnection is well-posed as in Definition~\ref{def:wellposed}.

Assume the stacked input data  $\{u(k)\}_{k=0}^{N-1}$ is PE of order $n_x+1$. Let $X_{N}, X_N^+, W_{N}, D_{N}, V_{N},$ and $E_{N}$ be data matrices, defined in \eqref{eq:data}, constructed from trajectories of the nominal LTI system $G$ excited by this stacked input. 

Then any $(P,\gamma,\epsilon)$ satisfy $L_{MB}(P,M,\gamma^2,\epsilon) \preceq 0$ if and only if they satisfy $L_{SSD}(P,M,\gamma^2,\epsilon) \preceq 0$.
Moreover, the interconnection $F_U(G,\Delta_\Phi)$ is internally stable and has $\|F_U(G,\Delta_\Phi)\|_{2\to 2} < \gamma$ if there exists $P\succ 0$ and  $\gamma > 0$, and $\epsilon>0$ such that either of these (equivalent) LMI conditions is satisfied.
\end{theorem}
\begin{proof}
Define $R := \bmtx X_{N}^\top & W_{N}^\top & D_{N}^\top \emtx^\top$.
The definitions of $L_{MB}$ in~\eqref{eq:LMB} and $L_{SSD}$ in~\eqref{eq:LSSD} combined with the dynamics of $G$ in \eqref{eq:LTInom} give:
\begin{align}
    \label{eq:LSSDandLMB}
    L_{SSD}(P,M,\gamma^2,\epsilon) = R^\top\, L_{MB}(P,M,\gamma^2,\epsilon)\, R.
\end{align}
It follows immediately that $L_{MB} \preceq 0$ implies $L_{SSD} \preceq 0$ (Proposition 8.1.2 (xii) of~\cite{bernstein09}).

Moreover, it follows from Corollary~2 (ii) of~\cite{WILLEMS2005325}, that if $u$ is PE of order $n_x+1$ then $R$ has full row rank. Hence the pseudoinverse of $R$ satisfies $R R^\dagger = I$ (by Prop 6.1.6 (xix) of \cite{bernstein09}).
Multiplying \eqref{eq:LSSDandLMB} on the right and left by $R^\dagger$ and its transpose yields:
\begin{align*}
    L_{MB}(P,M,\gamma^2,\epsilon) = (R^\dagger)^\top\, L_{SSD}(P,M,\gamma^2,\epsilon)\, R^\dagger.
\end{align*}
Again, it follows that $L_{SSD}\preceq 0$ implies $L_{MB}\preceq 0$. Thus the two LMI conditions are equivalent.

Finally, if $\exists \,P\succ 0$, $\gamma > 0$, and $\epsilon>0$ such that
$L_{MB}(P,M,\gamma^2,\epsilon) \preceq 0$ then, by
Lemma~\ref{lem:stabMB}, the interconnection is internally stable with induced gain less than $\gamma$. Feasibility of the equivalent condition $L_{SSD}(P,M,\gamma^2,\epsilon) \preceq 0$ also implies internal stability and the gain bound.
\end{proof}
\vspace{0.1in}

Note that the input data richness (PE condition) is needed to show that
$L_{SSD}\preceq 0$ implies $L_{MB}\preceq 0$, but not the reverse direction.
The gain bound  from the data-driven condition in Theorem~\ref{thm:stabMF} guarantees dissipativity from the observed data of length $N$. 
It matches the model-based SDP bound, in general, only when the PE condition in Definition~\ref{def:PE} is satisfied.  This PE condition implies that Hankel matrix $U_{(n_x+1),j}^p$ has at least $(n_x+1)n_u$ columns.  The total data length is $N=i+j-1$ with $i=(n_x+1)$ and $j\ge (n_x+1)n_u$. Thus the total data length must be at least $N \ge n_x n_u + n_x + n_u$. This ensures the data captures the full system behavior. 

\subsection{Input-Output Data-Driven Stability and Performance}
\label{sec:IOdataDriven}

In this subsection, we derive a data-driven stability and performance condition that does not require direct access to the state trajectory. Instead, a state sequence is reconstructed from measured input-output data using a deterministic subspace identification procedure. We follow the notation and formulations in \cite{overschee1996subspace}. The reconstructed state is equivalent to the true state up to an invertible change of coordinates. This allows the model-based dissipation inequality to be reformulated entirely in terms of measured input-output data. The nonlinear relation between $v$ and $w$ is handled separately through a static quadratic constraint.

Let the integers $i,j\ge 1$ be given. Define the extended observability matrix $\Gamma_i$ and lower block-triangular Toeplitz matrix $H_i$ as follows:
\begin{align*}
\Gamma_{i}&:=
\bmtx
C^\top &
(CA)^\top &
\cdots &
(CA^{i-1})^\top
\emtx^\top \in \R^{i n_y \times n_x} \\
H_i & := \bmtx D & 0 & \cdots & 0 \\
CB & D & \cdots & 0\\
\vdots & \vdots &  & \vdots \\
CA^{i-2}B & CA^{i-3}B & \cdots & D\emtx \in \R^{(in_y) \times (in_u)}.
\end{align*}
Also, define the block-Hankel matrix  using $\{u(k) \}_{k=i}^{\hat{N}-1}$ where $\hat{N}:=2i+j-1$:
\begin{align*}
    U_{i,j}^f := \bmtx u(i) & u(i+1) & \cdots & u(i+j-1)\\ u(i+1) & u(i+2) & \cdots & u(i+j)\\ \vdots & \vdots & & \vdots\\ u(2i-1) & u(2i) & \cdots & u(2i+j-2)
    \emtx \in \R^{(in_u) \times j}.
\end{align*}
Note that $U_{i,j}^f$ is defined similarly to $U_{i,j}^p$ in~\eqref{eq:pInHank}. The superscripts $p$ and $f$ denote that each column decomposes the input into past and future data.  For example, the first columns of $U_{i,j}^p$ and $U_{i,j}^f$ depend on $\{ u(k)\}_{k=0}^{i-1}$ and  $\{ u(k)\}_{k=i}^{2i-1}$, respectively.  
The lower right entry of $U_{i,j}^f$ corresponds to $u(\hat{N}-1)$. Hence constructing both $U_{i,j}^p$ and $U_{i,j}^f$ requires input data of length $\hat{N}$. The output block-Hankel matrices $Y_{i,j}^p$ and $Y_{i,j}^f \in \R^{(in_y) \times j}$ are defined similarly using the output data $\{y(k)\}_{k=0}^{N-1}$ and $\{ y(k)\}_{k=i}^{\hat{N}-1}$, respectively.

We state two supporting lemmas before stating the main input/output data-driven result. In the first supporting lemma, $\text{row}(M)$ denotes the row space of the matrix $M$.

\vspace{0.1in}
\begin{lemma}
\label{lem:assumption2rank}
Let stacked input data $\{u(k) \}_{k=0}^{\hat{N}-1}$ be given along with integers
$j\ge 1$ and $i>n_x$ such that $\hat{N}=2i+j-1$.  
Suppose $(A,C)$ of the nominal LTI system $G$ in \eqref{eq:LTInom} is observable. Assume the block-Hankel matrices constructed from the input/output data $(u,y)$ over $k=0,\ldots,\hat{N}-1$ satisfy
\begin{align}
    \operatorname{rank} 
    \left(
    \bmtx {U_{i,j}^p} \\ {Y_{i,j}^p} \\ {U_{i,j}^f} 
    \emtx
    \right)
    = 2in_u+n_x.
    \label{eq:rankHank}
\end{align}
Then $\operatorname{row}(U_{i,j}^f)\cap \operatorname{row}(X_j)=\{0\}$.
\end{lemma}
\vspace{0.1in}
\begin{proof}
The proof is given in Appendix~\ref{sec:proofLem2}.
\end{proof}
\vspace{0.1in}

Next, we use subspace identification results to reconstruct the state, up to a coordinate transformation, from input/output data. This requires notation for orthogonal and oblique projections following Sections~1.4.1 and~1.4.2
of~\cite{overschee1996subspace}. The orthogonal complement of a subspace $\mathcal{V}$ is denoted by $\mathcal{V}^\perp$. 
Let $A\in\R^{p\times j}$, $B\in\R^{q\times j}$, and
$C\in\R^{r\times j}$ be matrices with the same number of
columns. Define the orthogonal
projection matrices onto $\operatorname{row}(B)$ and
$\operatorname{row}(B)^\perp$ by
$\Pi_B:=B^\top(BB^\top)^\dagger B$,
and  $\Pi_B^\perp:=I_j-\Pi_B$, respectively.
Accordingly, the orthogonal projections of
$\operatorname{row}(A)$ onto $\operatorname{row}(B)$ and
$\operatorname{row}(B)^\perp$ are denoted by $A/B := A\Pi_B, \; A/B^\perp := A\Pi_B^\perp$,
respectively. Finally, the oblique projection of
$\operatorname{row}(A)$ along $\operatorname{row}(B)$ onto
$\operatorname{row}(C)$ is denoted by
$A/_{B}\, C := (A\Pi_B^\perp) \, (C\Pi_B^\perp)^\dagger C$.

Let $\mathcal{O}_i \in \R^{in_y\times j}$ denote the oblique projection of the future outputs $Y_{i,j}^f$ onto the past data $(U_{i,j}^p,Y_{i,j}^p)$ along the future inputs $U_{i,j}^f$.
Following Section~1.4.2 of~\cite{overschee1996subspace}, this oblique projection is given by
\begin{align}
\begin{split}
\mathcal{O}_i :=& Y_{i,j}^f /_{U_{i,j}^f} \bmtx U_{i,j}^p \\ Y_{i,j}^p \emtx = Y_{i,j}^f \Pi_{U_{i,j}^f}^{\perp} \left(\bmtx U_{i,j}^p \\ Y_{i,j}^p \emtx \Pi_{U_{i,j}^f}^{\perp}\right)^\dagger \bmtx U_{i,j}^p \\ Y_{i,j}^p \emtx,
\end{split}
\label{eq:obliqueO}
\end{align}
where
\begin{align}
\Pi_{U_{i,j}^f}^{\perp} :=& I - {U^f_{i,j}}^\top \left(U_{i,j}^f\, {U^f_{i,j}}^\top\right)^{\dagger} U_{i,j}^f.
\label{eq:obliquePi}
\end{align}
This oblique projection is used in the next subspace identification result.

\vspace{0.1in}
\begin{lemma}
\label{lem:subspace}
Assume the nominal LTI system $G$ in~\eqref{eq:LTInom} is minimal, i.e., $(A,B)$ is controllable and $(A,C)$ is observable. Let stacked input data $\{u(k) \}_{k=0}^{\hat{N}-1}$ be given along with integers
$j\ge 1$ and $i>n_x$ such that $\hat{N}=2i+j-1$.  Suppose that the following conditions hold:
\begin{enumerate}
    \item  The stacked input data  is PE of order $2i$ as in Definition~\ref{def:PE}. 
    
   \item  The block-Hankel matrices $(U_{i,j}^p, U_{i,j}^f, Y_{i,j}^p)$ constructed 
     from the input/output data $(w,d,v,e)$ over $k=0,\ldots,\hat{N}-1$ 
     satisfy the rank condition in \eqref{eq:rankHank}.
\end{enumerate}
Then $\mathcal{O}_{i}$ has rank $n_x$.  Moreover,
Let $\mathcal{O}_{i} = U_r \Sigma_r V_r^\top$ be the compact singular value decomposition of the oblique projection~\eqref{eq:obliqueO}, retaining the $n_x$ nonzero singular values. Define the reconstructed state data matrix by:
\begin{align}
\label{eq:projstate}
Z := \Sigma_r^{1/2}\, V_r^\top 
\in \R^{n_x \times j}.
\end{align}
Let the $j$ columns of this matrix be denoted by:
\begin{align*}
    Z := \bmtx z(i) & z(i+1) & \cdots, & z(i+j-1) \emtx.
\end{align*}
There exists an invertible matrix $T\in \R^{n_x \times n_x}$ such that $z(k)=T \, x(k)$
where $\{ x(k) \}_{k=i}^{i+j-1}$ is the state data from the nominal system $G$ in \eqref{eq:LTInom} driven by the stacked input from an initial condition $x(0)$. 

\end{lemma}
\begin{proof}
This result follows from Theorem 2 in 
\cite{overschee1996subspace}. Specifically, this theorem has two key assumptions. First, 
Theorem 2 in  \cite{overschee1996subspace}.
requires the input to be PE of order $2i$. This follows from condition 1) in the lemma statement.\footnote{Definition 5 in \cite{overschee1996subspace} states that the input is PE of order $2i$ if $U^p_{2i,j} (U^p_{2i,j})^\top$ is full rank equal to $2i$.  This is equivalent to $U^p_{2i,j}$ having full row rank equal to $2i$. Hence the definition of PE in \cite{overschee1996subspace} is equivalent to the definition of PE from~\cite{WILLEMS2005325} as used in our paper.}
Second, Theorem 2 in 
\cite{overschee1996subspace}. requires $\operatorname{row}(U_{i,j}^f)\cap \operatorname{row}(X_j)=\{0\}$. This follows condition 2) and  Lemma~\ref{lem:assumption2rank}. The second conclusion of Theorem 2 in 
\cite{overschee1996subspace} is that 
$\mathcal{O}_{i}$ has $n_x$ non-zero singular values. The state reconstruction $Z$ follows from conclusions 4 and 5 of Theorem 2 in 
\cite{overschee1996subspace}.
\end{proof}
\vspace{0.1in}

Note that  Theorem 2 in 
\cite{overschee1996subspace}
requires $\operatorname{row}(U_{i,j}^f)\cap \operatorname{row}(X_j)=\{0\}$.  However, state data is needed to verify this assumption.
Lemma~\ref{lem:assumption2rank} is used to verify
this assumption without state data. Specifically, this lemma only requires  the rank  condition
\eqref{eq:rankHank} involving input/output data. 
The columns of $Z$ correspond to the state data for the system  $(TAT^{-1}, TB, C T^{-1},D)$. In other words, $Z$ contains the state data reconstructed up to a similarity transformation.  

Next, we can  define stacked data input, output, and reconstructed state data similar to \eqref{eq:data}, but starting from time $k=i$ rather than $k=0$. If $j\ge 2$, then the data matrices with length $N=j-1$ are defined as:
\begin{align}
\label{eq:datai}
\begin{split}
Z_{i,N} &= \bmtx z(i) & z(i+1) & \cdots & z(i+N-1) \emtx \\
Z_{i,N}^+ &= \bmtx z(i+1) & z(i+2) & \cdots & z(i+N) \emtx. \\
W_{i,N}&= \bmtx w(i) & w(i+1) & \cdots & w(i+N-1) \emtx, \\
D_{i,N}&= \bmtx d(i) & d(i+1) & \cdots & d(i+N-1) \emtx, \\
V_{i,N}&= \bmtx v(i) & v(i+1) & \cdots & v(i+N-1) \emtx, \\
E_{i,N}&= \bmtx e(i) & e(i+1) & \cdots & e(i+N-1) \emtx.
\end{split}
\end{align}
All these matrices have $N=j-1$ columns.

Define the following matrix function $L_{IOD}: \Sym^{n_x} \times \Sym^{2m} \times \R_{>0} \times \R_{>0} \to \Sym^{N}$
using these data matrices:
\begin{align}
\label{eq:LIOD}
& L_{IOD}(P,M,\gamma^2,\epsilon):=\,
\epsilon \bmtx Z_{i,N} \\ W_{i,N} \\ D_{i,N} \emtx^\top \bmtx Z_{i,N} \\ W_{i,N} \\ D_{i,N} \emtx \\
\nonumber
& \hspace{0.2in}
- Z_{i,N}^\top\, P\, Z_{i,N} + {Z_{i,N}^+}^\top\, P\, Z_{i,N}^+\\
\nonumber
& \hspace{0.2in}
- \gamma^2\, D_{i,N}^\top\, D_{i,N}
 + E_{i,N}^\top\, E_{i,N}
+  \bmtx V_{i,N} \\ W_{i,N}\emtx^\top
M
  \bmtx V_{i,N} \\ W_{i,N}\emtx.
\end{align}
The subscript $IOD$ indicates that this will be part of the input-output data-driven formulation, which uses measured input-output data and does not require state measurements.

\vspace{0.1in}
\begin{theorem}
\label{thm:lifted_io_state}
Consider the interconnection $F_U(G,\Delta_\Phi)$, where the nominal LTI system $G$ in~\eqref{eq:LTInom} is minimal and $\Phi:\mathbb{R}^m\to\mathbb{R}^m$ is a memoryless nonlinearity satisfying the QC defined by $M$. Assume the interconnection is well-posed as in Definition~\ref{def:wellposed}.

Let stacked input data $\{u(k) \}_{k=0}^{\hat{N}-1}$ be given along with integers
$j\ge 2$ and $i=n_x+1$ such that $\hat{N}=2i+j-1$.
Assume the data satisfies the two conditions in Lemma~\ref{lem:subspace}.  Let $Z$ denote the state matrix reconstructed via the procedure in Lemma~\ref{lem:subspace}. 
Let $Z_{i,N}, Z_{i,N}^+, W_{i,N}, D_{i,N}, V_{i,N},$ and $E_{i,N}$ be data matrices, defined in \eqref{eq:datai} with $N=j-1$, constructed from trajectories of $G$ excited by this stacked input. Finally, assume the subset of data $\{u(k)\}_{k=i}^{i+N-1}$ is PE of order $n_x+1$.

Then the interconnection $F_U(G,\Delta_\Phi)$ is internally stable and satisfies $\|F_U(G,\Delta_\Phi)\|_{2\to 2} < \gamma$ if there exists $P \succ 0$, $\gamma > 0$ and $\epsilon>0$ such that
$L_{IOD}(P,M,\gamma^2,\epsilon) \preceq 0$.
\end{theorem}
\vspace{0.1in}
\begin{proof}
Let $L_{SSD}^i(P,M,\gamma^2,\epsilon) \preceq 0$ denote the SSD condition but with $N$ columns of data starting at $k=i$ rather than $k=0$, i.e. $X_N$ replaced by $X_{i,N}$, etc. 
By Lemma~\ref{lem:subspace}, there exists a nonsingular matrix $T\in \R^{n_x\times n_x}$ such that $z(k)=T\,x(k)$. Define $\tilde{\epsilon}:= \lambda_{\min}(T)^2 \epsilon$.  Then  the I/O data condition $L_{IOD}(P,M,\gamma^2,\epsilon) \preceq 0$ implies that
$L_{SSD}^i(T^\top P T,M,\gamma^2,\tilde{\epsilon}) \preceq 0$.
The system is time-invariant, so we can equivalently use the data starting at $k=i$ rather than $k=0$. In other words, $
L_{SSD}^i(T^\top P T,M,\gamma^2,\tilde{\epsilon}) \preceq 0$  if and only if $L_{SSD}(T^\top P T,M,\gamma^2,\tilde{\epsilon}) \preceq 0$.
Theorem~\ref{thm:stabMF} holds because the subset of stacked input data $\{u(k)\}_{k=i}^{i+N-1}$ is PE of order $n_x+1$. It follows that the system is internally stable and$\|F_U(G,\Delta_\Phi)\|_{2\to 2} < \gamma$.
\end{proof} 
\vspace{0.1in}

Theorem~\ref{thm:lifted_io_state} has two PE requirements.  The first is that the
stacked input data $\{u(k) \}_{k=0}^{\hat{N}-1}$ 
must be PE of order $2i$ with $i=n_x+1$. The second is that the subset of data from $k=i$ to $k=i+N-1$ must be PE of order $n_x+1$. The upper index of this data subset is equivalent to $k=i+N-1=\hat{N}-i-1$. Hence the data subset in the second requirement trims the first and last $i$ samples of data. It is possible to construct input sequences that satisfy the first PE requirement but not the second requirement.\footnote{Consider the case with $n_u=2$, $i=2$, and $\hat{N}=11$. Define the following sequence for $\{u(k)\}_{k=0}^{\hat{N}-1}$:
\begin{align*}
\left\{ \bmtx 0 \\ 0 \emtx, 
\bmtx 0 \\ 1 \emtx, 
\bmtx 4 \\ 0 \emtx, 
\bmtx 10 \\ 0 \emtx, 
\bmtx 30 \\ 0 \emtx, 
\bmtx 100 \\ 0 \emtx 
\bmtx 354 \\ 0 \emtx,
\bmtx 1300 \\ 0 \emtx, 
\bmtx 4890 \\ 0 \emtx, 
\bmtx 0 \\ 1 \emtx,
\bmtx 0 \\ 0 \emtx \right\} 
\end{align*}
Consider the Hankel matrix for this sequence  $U_{2i,j}^p$ with $j=8$. This matrix has $2in_u=8$ rows and has (full) rank equal to 8. Hence the data is PE of order $2i=4$. Next, consider the data subset $\{u(k)\}_{k=i}^{\hat{N}-i-1}$ with  $\hat{N}-2i=7$ elements, i.e. the first and last $i=2$ entries have been removed. Consider the  Hankel matrix for this subset $U_{i,j}^p$ with $j=6$. This matrix has $in_u=4$ rows but the rank is only 2. Hence this subset is not PE of order $i$.} This is the reason that the second PE requirement is needed in Theorem~\ref{thm:lifted_io_state}.

The IOD formulation generally requires a longer data trajectory than the SSD formulation. We need the stacked input data $\{u(k) \}_{k=0}^{\hat{N}-1}$ to be PE of order $2i$ with $i=(n_x+1)$ so that the state can be reconstructed.   This PE condition implies that Hankel matrix $U_{2i,j}^p$ has at least $j\ge 2in_u$ columns.  Thus the total data length must satisfy $\hat{N}\ge 2i+2in_u-1$ with $i=(n_x+1)$. Simplifying yields the following bound on the required data: 
$\hat{N} \ge 2(n_x n_u + n_x + n_u)+1$.
We also need the stacked matrix in~\eqref{eq:rankHank} to have at least
$j \ge 2in_u+n_x$ columns to satisfy the rank condition.  This is an additional $n_x$ columns of data as compared to the bound implied by the PE condition.  Hence the total data length must satisfy the stronger condition $\hat{N} \ge 2n_xn_u + 3n_x + 2n_u+1$.

In comparison, the SSD formulation requires a data length of at least $n_xn_u+n_x+n_u$. Thus, eliminating the need for state
measurements comes at the cost of a larger minimum data requirement. These bounds are only necessary dimensional conditions; longer trajectories may be required for the
corresponding rank conditions to hold in practice.

\section{Numerical examples}
\label{sec:example}

In this section, we use the conditions in Theorem~\ref{thm:stabMF} and Theorem~\ref{thm:lifted_io_state} to analyze the induced-$\ell_2$ gain for an example system.
Consider the closed-loop system in Figure~\ref{fig:closedloop_parasitic}, formed by an LTI plant $G_p$ and a controller $G_c$. The closed-loop includes  a parasitic nonlinearity $\Delta_\beta$. 
Existing data-driven stability and performance methods that consider only linear systems cannot capture the effect of the nonlinearity. We will use the SSD and IOD conditions to analyze this system.

The plant $G_p$ is  a discrete-time, second-order LTI system:
\begin{align*}
\begin{split}
    x_p(k+1)& = A_p\, x_p(k) + B_{pw}\, w(k) + B_{pd}\, d(k), \\
    e(k)& = -C_p\, x_p(k),
\end{split}
\end{align*}
with the following state matrices
\begin{align*}
\begin{split}
    A_p &= \bmtx 0.90 & 0.08 \\ -0.03 & 0.82 \emtx, \quad
    B_{pw} = 0.11 \times \bmtx 1 & 0.2 \\ 0.1 & 1 \emtx, \\
    B_{pd} &= \bmtx 0.15 & 0 \\ 0 & 0.15 \emtx, \quad C_p=I_2.
\end{split}
\end{align*}
The controller $G_c$ is also a discrete-time, second-order LTI system:
\begin{align*}
\begin{split}
    x_c(k+1)& = A_c\, x_c(k) + B_c\, e(k), \\
    v(k)& = C_c\, x_c(k) + D_c\, e(k),
\end{split}
\end{align*}
with the following state matrices
\begin{align*}
\begin{split}
    A_c &= 0.98 \times I_2, \quad B_c = 0.1 \times I_2, \\
    C_c &= \bmtx 0.25 & 0 \\ 0 & 0.35 \emtx, \quad D_c = \bmtx 0.40 & 0 \\ 0 & 0.60 \emtx.
\end{split}
\end{align*}
The state of the nominal closed-loop LTI system is
$x=\bmtx x_p^\top & x_c^\top\emtx^\top$, and hence $n_x=4$.

Finally, the parasitic nonlinearity  $\Delta_\beta:\ell_{2e}^2\to \ell_{2e}^2$  is defined by a static nonlinearity  $\Phi_\beta:\R^2 \to \R^2$ at each point in time: $w(k)=\Phi_\beta(v(k))$ for all $k\ge 0$. Moreover, the static nonlinearity is a repeated function, i.e., there  exists $\phi_\beta:\R\to \R$  such that $w=\Phi_\beta(v)$ is defined entrywise by $w_i=\phi_\beta(v_i)$ for $i=1,\ldots,m$. We assume the nonlinearity is sector bounded by $[1-\beta,1+\beta]$: 
\begin{align*}
( (1+\beta) v - w ) \, (w - (1-\beta) v) \ge 0
\;\;
\forall v\in\R, \; w=\phi_\beta(v).
\end{align*}
In other words, the graph of $\phi_\beta$ lies between lines of slope $(1-\beta)$ and $(1+\beta)$.  Hence  $\beta>0$ quantifies the amount of nonlinearity around a line of slope equal to one. Larger values of $\beta$ correspond to a wider sector and allow the nonlinearity to have a stronger effect on the system.

It is well known (Section 8.1 of \cite{boyd1994linear} or Section VI.I of \cite{megretski97}) that $\Phi_\beta$ satisfies the QC defined by any matrix in the following set: 
\begin{align*}
 \mathcal{M} := \left\{   M:= \bmtx -(1-\beta^2)\Lambda & \Lambda\\ \Lambda & -\Lambda \emtx \, : \,
\Lambda = \bmtx \lambda_1 & 0 \\ 0 & \lambda_2 \emtx\succeq 0
 \right\}.
\end{align*}

\begin{figure}[t]
\centering
\scalebox{0.82}{
\begin{picture}(205,95)(15,10)
\thicklines
\put(30,50){\circle{8}}
\put(32,35){$-$}
\put(8,53){$0$}
\put(11,50){\vector(1,0){15}}
\put(42,53){$e$}
\put(34,50){\vector(1,0){21}}
\put(55,35){\framebox(30,30){$G_c$}}
\put(92,53){$v$}
\put(85,50){\vector(1,0){28}}
\put(113,35){\framebox(32,30){$\Delta_\beta$}}
\put(152,53){$w$}
\put(145,50){\vector(1,0){28}}
\put(173,36){\framebox(40,40){$G_p$}}
\put(152,85){$d$}
\put(154,82){\line(0,-1){15}}
\put(154,67){\vector(1,0){19}}
\put(213,56){\line(1,0){15}}
\put(228,56){\line(0,-1){41}}
\put(228,15){\line(-1,0){198}}
\put(30,15){\vector(0,1){31}}
\end{picture}
}
\caption{Closed-loop system with an LTI plant $G_p$, a controller $G_c$, and a parasitic nonlinearity $\Delta_\beta$ between them.}
\label{fig:closedloop_parasitic}
\end{figure}
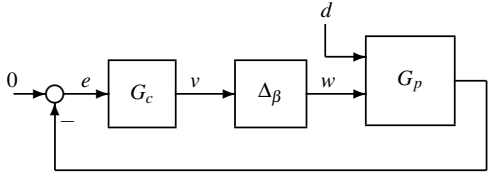

\begin{figure}[h!]
  \centering
  \includegraphics[width=0.4\textwidth]{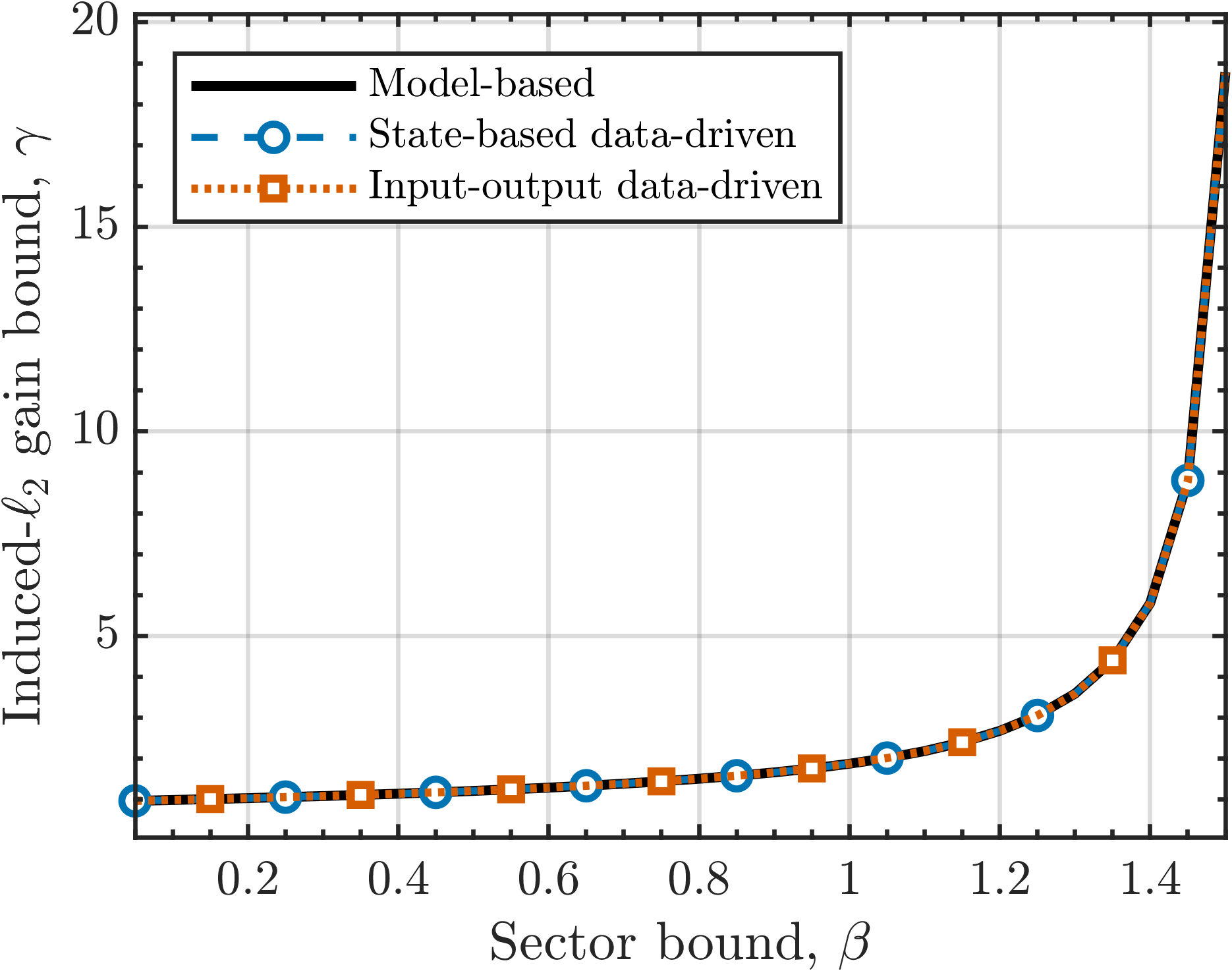}
  \caption{Induced-$\ell_2$ gain bounds as a function of $\beta$. The model-based, state-based data-driven, and input-output data-driven formulations recover identical gain bounds.}
\label{fig:GammaBeta}
\end{figure}

For both data-driven formulations, the trajectories are generated
in closed loop by exciting only the input
$d$, while signal $w$ is determined by the feedback
interconnection. Specifically, each entry of $d$ is generated
as
\begin{align*} d_r(k)=g_r(k)+0.2\sin\!\left((0.11+0.017r)k+0.31r\right),
\end{align*}
for $r=1,\ldots,n_d$, where $g_r(k)\sim\mathcal{N}(0,1)$ are independent Gaussian
samples generated using \texttt{randn} in MATLAB. For each value
of $\beta$, the closed-loop data are generated using the repeated
memoryless nonlinearity
\begin{align*}
 w_r(k)=v_r(k)+\beta v_r(k)\sin(v_r(k)),
\end{align*}
for $r=1,\ldots,m$. Since $\sin(v_r)\in[-1,1]$, it follows that
$1-\beta\leq\frac{w_r}{v_r}\leq 1+\beta$
for $v_r\neq0$, while the corresponding sector inequality also
holds at $v_r=0$. Therefore, $\Phi_\beta$ belongs to the sector
$[1-\beta,1+\beta]$, consistent with the uncertainty class used
in the analysis. 
After collecting the data, we verify that
$u(k)=\bmtx w(k)^\top & d(k)^\top \emtx^\top$
is PE of order $n_x+1$ for
Theorem~\ref{thm:stabMF} and of order $2(n_x+1)$ for
Theorem~\ref{thm:lifted_io_state}, while also satisfying the rank condition in~\eqref{eq:rankHank}.

For the input-output data-driven condition in Theorem~\ref{thm:lifted_io_state}, only the input-output trajectories are used, and a state sequence is reconstructed through the deterministic subspace-identification procedure in Lemma~\ref{lem:subspace}. The conditions in Lemma~\ref{lem:subspace} are verified using similar PE signals.  The required data lengths are $24$ for Theorem~\ref{thm:stabMF} and $53$ for Theorem~\ref{thm:lifted_io_state}. Both of these data sets satisfy the minimum data length assumptions described after Theorems \ref{thm:stabMF}
 and \ref{thm:lifted_io_state}.

The optimal (minimal) gain bound $\gamma$ is obtained by solving the following semidefinite program (SDP) with $\epsilon=10^{-8}$:
\begin{align*}
    \min_{M\in \mathcal{M},P\succ 0, \gamma^2>0} \gamma^2 
    \,\, \mbox{ s.t. } \,\,
    LMI_{\bullet}(P,M,\gamma^2,\epsilon) \preceq 0
\end{align*}
The $\bullet$ denotes that we solve this SDP using the $MB$, $SSD$, and $IOD$ LMI constraints in Lemma~\ref{lem:stabMB}, Theorem~\ref{thm:stabMF}, and Theorem~\ref{thm:lifted_io_state}, respectively. We solve this SDP using CVX \cite{cvx14} as the interface and MOSEK \cite{MOSEK} as the underlying solver. The analysis is repeated for $\beta \in \{0.1,0.15,\ldots,1.5\}$. We also solve with $\beta=0$ as the linear baseline. 

Figure~\ref{fig:GammaBeta} shows the results of the induced gain-bounds $\gamma$ as a function of sector size $\beta$.
As expected, the data-driven conditions SSD and IOD both recover the same optimal induced-$\ell_2$ gain bounds as the model-based condition MB over the full range of sector bounds.  As shown in Figure~\ref{fig:GammaBeta}, the induced-$\ell_2$ gain bound changes noticeably as the sector becomes larger. The case $\beta=0$ corresponds to the baseline linear system with $w=v$. In this case, the gain bound reduces to the $H_\infty$ norm of the transfer matrix from $d$ to $e$, which is $\gamma=0.951$. The gain bounds from the MB, SSD, and IOD conditions converge to this linear gain as $\beta\to 0$. The average computation times for each bound in the $\beta$-grid are $0.167$, $0.175$, and $0.196$ seconds for the MB (Lemma~\ref{lem:stabMB}), SSD (Theorem~\ref{thm:stabMF}), and IOD (Theorem~\ref{thm:lifted_io_state}) conditions, respectively.

\section{Conclusions}

This paper developed two data-driven stability and performance conditions for discrete-time Lurye systems with nonlinearities characterized by known quadratic constraints (QCs). The first condition uses measured input, output, and state trajectories of the nominal system. It is equivalent to the model-based condition under a persistency of excitation assumption on the input data. The second formulation requires only input/output data. The state is reconstructed, up to a similarity transformation, using deterministic subspace-identification methods. The resulting conditions were formulated as convex semidefinite programs and validated through a numerical example. Future work includes extending the framework to noisy trajectory data and incorporating dynamic integral quadratic constraints (IQCs). We can also merge our formulation with recent complementary results that assume the nominal dynamics are known but use measured data to learn QCs. This would yield a fully data-driven framework for Lurye systems.


\section{Acknowledgments}

The authors acknowledge AFOSR Grant \#FA9550-23-1-0732 for funding of this work.

\bibliographystyle{myIEEEtran}
\bibliography{references} 

\appendix 
\subsection{Proof of Lemma~\ref{lem:assumption2rank}}
\label{sec:proofLem2}

The data matrices satisfy the following relation (see Equation~(2.5) of~\cite{overschee1996subspace}):
\begin{align}
\label{eq:IODataRelation}
     Y_{i,j}^p = \Gamma_{i} X_j +  H_i U_{i,j}^p 
\end{align}
This expression gives
\begin{align*}
    \bmtx U_{i,j}^p \\ Y_{i,j}^p \emtx = \bmtx I & 0 \\ H_i & \Gamma_{i}\emtx \bmtx U_{i,j}^p \\ X_j \emtx.
\end{align*}
Therefore, by item~(c) of Section~0.4.5 in~\cite{horn2012matrix},
\begin{align}
\label{eq:upboundUY}
    \operatorname{rank}\left(\bmtx U_{i,j}^p \\ Y_{i,j}^p \emtx\right) \le \operatorname{rank}\left(\bmtx U_{i,j}^p \\ X_j \emtx\right) \le
    in_u+n_x.
\end{align}
The second inequality follows because the rank cannot exceed the number of rows. Next, apply the subspace intersection theorem in Section~0.1.7 of~\cite{horn2012matrix} to obtain:
\begin{align}
\nonumber
\operatorname{rank}\left(\bmtx U_{i,j}^p \\ Y_{i,j}^p \\ U_{i,j}^f \emtx\right) 
& = 
\operatorname{rank}\left(\bmtx U_{i,j}^p \\ Y_{i,j}^p \emtx\right) \\
& + \operatorname{rank}(U_{i,j}^f)
 -\operatorname{dim}\left(\mathcal{V} \right).
\label{eq:dimFormula}
\end{align}
where $\operatorname{dim}\left(\mathcal{V} \right)$ is the dimension of the following subspace:
\begin{align*}
 \mathcal{V}:=    \operatorname{row}\left(\bmtx U_{i,j}^p \\ Y_{i,j}^p \emtx\right) \cap \operatorname{row}\left(U_{i,j}^f\right)
\end{align*}
The left side of \eqref{eq:dimFormula}
 is equal to $2in_u+n_x$ by assumption.
The first term on the right side is $\le in_u+n_x$ by
\eqref{eq:upboundUY}. The second term on the right side satisfies $\operatorname{rank}(U_{i,j}^f) \le in_u$ because the rank cannot exceed the number of rows. 
Combining these facts with \eqref{eq:dimFormula}
gives $\operatorname{dim}\left(\mathcal{V} \right) \le 0$.
Hence the dimension of this intersection is identically equal to zero so that $\mathcal{V} = \{ 0\}$.

To complete the proof, note that $(A,C)$ is observable and $i>n_x$, both by assumption. Hence the extended observability matrix $\Gamma_i$ has full column rank $n_x$ (see  Section~6.2.1 of~\cite{kailath1980linear}). Therefore, the pseudoinverse of $\Gamma_i$ satisfies $\Gamma_i^\dagger \Gamma_i = I$ by Proposition~6.1.6 (xviii) of~\cite{bernstein09}.
Multiplying \eqref{eq:IODataRelation}
on the left by $\Gamma_i^\dagger$ and rearranging yields
\begin{align*}
    X_j 
    = \bmtx - \Gamma_{i}^\dagger H_i & \Gamma_{i}^\dagger\emtx \bmtx U_{i,j}^p \\ Y_{i,j}^p \emtx.
\end{align*}
Thus each row of $X_j$ is a linear combination of rows from the stacked past input/output data. It follows that
\begin{align*}
\label{eq:rowSubset}
    \operatorname{row}(X_j) \subseteq \operatorname{row}\left(\bmtx U_{i,j}^p \\ Y_{i,j}^p \emtx\right).
\end{align*}     
This fact combined with $\mathcal{V} = \{0\}$ implies
\begin{align*}
    \operatorname{row}(X_j) \cap \operatorname{row}(U_{i,j}^f) = \{0\}.
\end{align*}
\hfill $\blacksquare$

\end{document}